\documentclass[aps,pre,showpacs,amsmath,amssymb,twocolumn]{revtex4-1}
\usepackage{amsthm}
\usepackage{dsfont}
\usepackage[dvips]{graphicx}
\usepackage[colorlinks,linkcolor=blue,anchorcolor=blue,citecolor=blue,filecolor=blue,menucolor=blue,urlcolor=blue]{hyperref}
\usepackage{breakurl}

\newcommand\ZZ{{\mathds{Z}}}
\newcommand\RR{{\mathds{R}}}
\newcommand\CC{{\mathds{C}}}

\newtheorem{proposition}{Proposition}
\newtheorem{corollary}{Corollary}

\begin{document}

\title{Energy landscape analysis of the two-dimensional nearest-neighbor \texorpdfstring{$\phi^4$}{phi4} model}

\author{Dhagash Mehta}
\email{dbmehta@syr.edu}
\affiliation{Department of Physics, Syracuse University, Syracuse, NY 13244, USA}

\author{Jonathan D.\ Hauenstein}
\email{jhauenst@math.tamu.edu}
\affiliation{Department of Mathematics, Texas A\&M University, College Station, TX 77843-3368, USA}

\author{Michael Kastner}
\email{kastner@sun.ac.za}
\affiliation{National Institute for Theoretical Physics (NITheP), Stellenbosch 7600, South Africa}
\affiliation{Institute of Theoretical Physics,  University of Stellenbosch, Stellenbosch 7600, South Africa}

\date{\today}

\begin{abstract}
The stationary points of the potential energy function of the $\phi^4$ model on a two-dimensional square lattice with nearest-neighbor interactions are studied by means of two numerical methods: a numerical homotopy continuation method and a globally-convergent Newton-Raphson method. We analyze the properties of the stationary points, in particular with respect to a number of quantities that have been conjectured to display signatures of the thermodynamic phase transition of the model. Although no such signatures are found for the nearest-neighbor $\phi^4$ model, our study illustrates the strengths and weaknesses of the numerical methods employed.
\end{abstract}


\maketitle

\section{Introduction}
The stationary points of the potential energy function or other classical energy functions can be employed to calculate or estimate certain physical quantities of interest. Well-known examples include transition state theory or Kramers's reaction rate theory for the thermally activated escape from metastable states, where the barrier height (corresponding to the difference between potential energies at certain stationary points of the potential energy function) plays an essential role. More recently, a large variety of related techniques has become popular under the name of {\em energy landscape methods} \cite{Wales}, allowing for applications to many-body systems as diverse as metallic clusters, biomolecules and their folding transitions, or glass formers undergoing a glass transition.

In the late 1990s it was observed that properties of stationary points of the potential energy function $V$, i.e.\ configuration space points $q^\text{s}$ satisfying ${\mathrm d}V(q^\text{s})=0$, reflect in dynamical and statistical physical quantities simultaneously and show pronounced signatures near a phase transition \cite{CaPeCo00}. This observation sparked quite some research activity, reviewed in Ref.\ \cite{Kastner08}, including a theorem by Franzosi and Pettini asserting that, at least for a certain class of models, stationary points with $V(q^\text{s})/N=v_\text{c}$ are indispensable for the occurrence of an equilibrium phase transition at potential energy $v_\text{c}$ \cite{FraPe04,*FraPeSpi07}. This theorem requires a number of conditions to be satisfied: The potential energy function $V$ has to have the Morse property, has to be smooth, confining, and of short-range (see \cite{FraPe04,*FraPeSpi07} for a complete list of conditions and their definitions). At the time when these papers were published, one might have still hoped that some of the conditions on $V$ were merely technical, but not essential for the result. However, it became clear soon that the result can not be extended to long-range interacting models \cite{BaroniMSc,*GaSchiSca04,*AnAnRuZa04}, nor to non-confining potentials \cite{Kastner04}: These classes of potentials comprise cases which are particularly amenable to analytic calculations, and a direct relation between phase transitions and stationary points of $V$ could be ruled out through exactly solvable counterexamples.

Originally, the incentive for the study reported in the present article was to investigate the stationary points of a model that satisfies all the conditions required by Franzosi and Pettini \cite{FraPe04,*FraPeSpi07}. This is not an easy task, as in this class there are no exactly solvable models which have a phase transition \footnote{Note that most exactly solvable models with short-range interactions, like for example the two-dimensional nearest-neighbor Ising model, have a discrete configuration space and the energy landscape techniques we are interested in do not apply.}. As a model to study, we then opted for the nearest-neighbor $\phi^4$ model on a two-dimensional square lattice. This model, though not exactly solvable, appears to be relatively simple. Moreover, results on the stationary points of its long-range version were known and readily available for comparison \cite{BaroniMSc,*GaSchiSca04,*AnAnRuZa04}. The potential of the two-dimensional nearest-neighbor $\phi^4$ model is smooth, confining, and of short range. Moreover it has the Morse property for almost all values of the coupling constants (see Appendix \ref{s:appendix} for a proof) and therefore satisfies all requirements of Franzosi and Pettini's theorem.

Much to our surprise, we found that all stationary points $q^\text{s}$ of the potential energy function $V$ have non-positive potential energies, i.e., $V(q^\text{s})\leq0$. From this observation, one can conclude that the result of Franzosi and Pettini, allegedly proven in Ref.\ \cite{FraPe04}, is false. Furthermore, a numerical method put forward in Ref.\ \cite{FraPeSpi00} and applied to the very same two-dimensional $\phi^4$ model yields incorrect results. These findings, and a discussion of their implications, have been published in a Letter \cite{KastnerMehta11}. The non-positivity of the stationary energies $V(q^\text{s})$ was established in that Letter analytically, supported by results obtained with two different numerical methods. The main purpose of the present article is to give a detailed account of these numerical methods and to present a more detailed analysis of the properties of the stationary points of the two-dimensional nearest-neighbor $\phi^4$ model.

In Sec.\ \ref{s:model}, this model is introduced and some of its thermodynamic properties are reviewed. In Sec.\ \ref{s:homotopy}, the first of the numerical methods, namely homotopy continuation, is discussed. It is an algebraic-geometrical technique devised to obtain {\em all}\/ isolated stationary points of a given system of multivariate polynomial equations, but is restricted to fairly small lattice sizes. We have applied this method to square lattices of sizes $3\times3$ and $4\times4$. The stationary points obtained are analyzed with respect to their number, potential energies, indices, and Hessian determinants in Sec.\ \ref{s:properties}. The second numerical method, discussed in Sec.\ \ref{s:NR}, makes use of a globally convergent version of the Newton-Raphson algorithm for searching the zeros of a real-valued function. It can be applied to larger lattice sizes, but provides in general only a subset of the stationary points. We summarize and discuss our findings in the concluding Sec.\ \ref{s:conclusions}.

\section{Two-dimensional nearest-neighbor \texorpdfstring{$\phi^4$}{phi4} model}
\label{s:model}

On a finite square lattice $\Lambda\subset\ZZ^2$ consisting of $N=L^2$ sites, a real degree of freedom $\phi_i$ is assigned to each lattice site $i\in\Lambda$. By $\mathcal{N}(i)$ we denote the subset of $\Lambda$ consisting of the four nearest-neighboring sites of $i$ on the lattice under the assumption of periodic boundary conditions. The potential energy function of this model is given by
\begin{equation}\label{e:V}
V(q)=\sum_{i\in\Lambda}\Biggl[\frac{\lambda}{4!}q_i^4-\frac{\mu^2}{2}q_i^2+\frac{J}{4}\sum_{j\in\mathcal{N}(i)}(q_i-q_j)^2\Biggr],
\end{equation}
where $q=(q_1,\dots,q_N)$ denotes a point in configuration space $\Gamma=\RR^N$
\footnote{Our definition of $V$ coincides with the one in Ref.\ \cite{FraCaSpiPe99}, but differs from \cite{FraPeSpi00} by a factor $1/6$ in the quartic term. Judging from the critical temperatures and energies reported in the latter, as well as from their reference to \cite{FraCaSpiPe99}, we assume that there is a misprint in Ref.\ \cite{FraPeSpi00}. For the main conclusions in Ref.\ \cite{KastnerMehta11} and in the present article, the precise values of any of the constants are not crucial.}%
.
The parameter $J>0$ determines the coupling strength between nearest-neighboring sites and the parameters $\lambda,\mu>0$ characterize a local double-well potential each degree of freedom is experiencing.

In the thermodynamic limit $N\to\infty$ this model is known to undergo, at some critical temperature $T_\text{c}$, a continuous phase transition, in the sense that the configurational canonical free energy
\begin{equation}
f(T)=-\lim_{N\to\infty}\frac{T}{N}\ln\int_\Gamma {\mathrm d}^N\!q\;{\mathrm e}^{-V(q)/T}
\end{equation}
is nonanalytic at $T=T_\text{c}$. The transition is from a ``ferromagnetic'' phase with nonzero average particle displacement to a ``paramagnetic'' phase with vanishing average displacement (see \cite{MilchevHeermannBinder86} for more details as well as for Monte Carlo results).

Since we are interested in whether, and how, the phase transition reflects in the properties of the potential energy landscape, it is more adequate for our purposes to compare not to $T_\text{c}$, but to the critical potential energy per lattice site, $v_\text{c}$, of the transition \cite{Kastner06}. Both quantities are unambiguously related to each other in the thermodynamic limit via the caloric curve $v(T)$. This is true independently of the statistical ensemble used, as these ensembles are known to be equivalent for short-range models like the one we are studying \cite{*[{}] [{, Chapter 2.4.}] Ruelle}.


The critical potential energy $v_\text{c}$ is less frequently studied, in fact the only data we could find in the literature are from Monte Carlo simulations of fairly small system sizes $N=20\times20$ in Ref.\ \cite{FraCaSpiPe99}, with parameter values $\lambda=3/5$, $\mu^2=2$, and $J=1$. We use the same values of $\lambda$ and $\mu^2$ in the following, but will show results for a range of couplings $J$. Since the value of $v_\text{c}$ is a crucial benchmark when relating our stationary point analysis to the phase transition of the $\phi^4$ model, we have performed standard Metropolis Monte Carlo simulations for somewhat larger system sizes up to $128\times128$ and $10^7$ lattice sweeps. 
Some of the Monte Carlo results have already been reported in Ref.\ \cite{KastnerMehta11}. From these plots one can read off a critical potential energy per lattice site of roughly $v_\text{c}\approx2.2$ for coupling $J=1$. A more precise value or an estimate of the statistical error could be obtained by more extensive Monte Carlo simulations and/or a finite-size scaling analysis of the data, but the results as they are will be sufficient for our purposes. We have determined $v_\text{c}$ also for several other couplings, and the results are displayed in Fig.\ \ref{f:vc_vs_J}.
\begin{figure}\center
\includegraphics[width=0.7\linewidth]{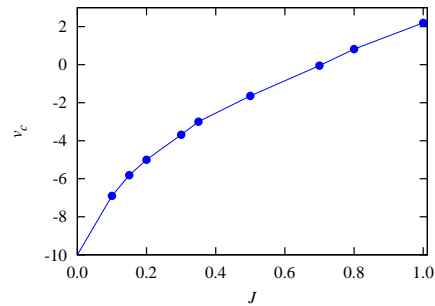}
\caption{\label{f:vc_vs_J}
(Color online) Monte Carlo results for the critical potential energy $v_c$ as a function of the coupling $J$ for the two-dimensional $\phi^4$ model \eqref{e:V} with $\lambda=3/5$ and $\mu^2=2$. System sizes up to $L=128$ have been used to obtain the estimates. The line connecting the data points is plotted as a guide to the eye.
}
\end{figure}

\section{Numerical polynomial homotopy continuation method}
\label{s:homotopy}
The idea behind numerical continuation methods is to first find the solutions of a simple system of equations which shares several important features with the given system. Then, in a second step, starting from these solutions one continues them towards the given system in a systematic way. Homotopy continuation methods have been around already for several decades \cite{Roth:62,AllgowerGeorg}. With more recent machinery like the numerical polynomial homotopy continuation (NPHC) method used in the present article, the method is guaranteed to find all isolated solutions of systems of polynomial equations \cite{SW:95,Li:2003}.

We consider a system of $m$ polynomial equations
\begin{equation}\label{e:poly}
P(q)=\begin{pmatrix}p_{1}(q)\\\vdots\\p_{m}(q)\end{pmatrix}=0
\end{equation}
in the variables $q=(q_{1},\dots,q_{m})^{T}$, and we assume that all solutions of \eqref{e:poly} are isolated. Then B\'ezout's Theorem (see Chapter 8 of \cite{SW:95}) asserts that a system of $m$ polynomial equations in $m$ variables has at most $\prod_{i=1}^m d_i$ isolated solutions where $d_i$ is the degree of the $i$th polynomial. This bound is called the \emph{classical B\'ezout bound}, and it is known to be sharp for generic systems [i.e., for generic values of the coefficients of the polynomials $p_i(q)$].

The continuation of solutions is formally described by the homotopy
\begin{equation}\label{e:homo}
H(q,t) = P(q) (1-t) + \gamma t S(q),
\end{equation}
where $\gamma$ is a complex number and
\begin{equation}\label{e:polystart}
S(q)=\begin{pmatrix}s_{1}(q)\\\vdots\\s_{m}(q)\end{pmatrix}=0
\end{equation}
is again a system of $m$ polynomial equations. Varying the parameter $t\in[0,1]$, $H$ can be deformed from the {\em start system} $H(q,1)=\gamma S(q)$ at $t=1$ into the polynomial system of interest, $H(q,0)= P(q)$ at $t=0$. The following conditions have to be satisfied in order to guarantee that all solutions of $P$ can be computed from this homotopy:
\begin{enumerate}
\item The solutions of $S(q)=0$ can be computed.
\item The number of solutions of $S(q)=0$ satisfies the classical B\'ezout bound for $P(q)=0$ as an equality.
\item The solution set of $H(q,t)=0$ for $t\in(0,1]$ consists of a finite number of smooth paths, called {\em homotopy paths}, which are parameterized by $t$.
\item Every isolated solution of $H(q,0)=P(q)=0$ can be reached by some path originating at a solution of $H(q,1)=\gamma S(q)=0$.
\end{enumerate}
Satisfying the first two criteria hinges on a suitable choice of the start system $S$. Criteria (iii) and (iv) are guaranteed to be satisfied based on the genericity of the constant $\gamma$ in \eqref{e:homo}. Theorem 8.4.1 of \cite{SW:95} states that these criteria hold for all but finitely many $\gamma$ on the unit circle.

The start system $S(q)=0$ can, for example, be taken to be
\begin{equation}\label{eq:Total_Degree_Homotopy}
S(q)=\begin{pmatrix}
q_{1}^{d_{1}}-1\\
\vdots\\
q_{m}^{d_{m}}-1
\end{pmatrix}=0,
\end{equation}
where $d_{i}$ is the degree of the $i^{th}$ polynomial of $P(q)=0$. The system \eqref{eq:Total_Degree_Homotopy} is easy to solve and guarantees that the total number of start solutions is $\prod_{i=1}^{m}d_{i}$ and all solutions are nonsingular.

Each homotopy path, starting at a solution of $S(q)=0$ at $t = 1$, is tracked to $t = 0$ using a path tracking algorithm, e.g., Euler predictor and Newton corrector methods. There are a number of freeware packages well-equipped with path trackers such as PHCpack \cite{Ver:99}, HOM4PS2 \cite{Li:03}, and Bertini \cite{BHSW06}. We used the latter one to get the results in this paper. Tracking the solutions to $t=0$, the set of endpoints of these homotopy paths is the set of all solutions to $P(q)=0$. Since each homotopy path can be tracked independently, NPHC is inherently parallelizable.

The set of real solutions can be obtained from the set of complex solutions by considering the imaginary part of the solutions (typically, up to a numerical tolerance). We remark that the approach of \cite{HS:12} implemented in alphaCertified \cite{alphaCertified} can be used to certify the reality or non-reality of a nonsingular solution given a numerical approximation of the solution. The ability to compute all complex solutions, and thus all real solutions, distinguishes the NPHC method from most other methods. Due to the power of the NPHC method, it has recently found several applications in theoretical physics \cite{Mehta:2009,*Mehta:2009zv,*Mehta:2011xs,*Mehta:2011wj}. 

To find the stationary points of the nearest neighbor $\phi^4$ model, we need to solve its stationary equations, i.e.,
\begin{equation}\label{e:root}
\begin{pmatrix}\frac{\partial V}{\partial q_1}(q^\text{s})\\\vdots\\\frac{\partial V}{\partial q_N}(q^\text{s})\end{pmatrix}=0
\end{equation}
with $q^\text{s}\equiv(q^\text{s}_1,\dotsc,q^\text{s}_N)\in\CC^N$. Since \eqref{e:root} is a system of $N$ coupled third-order polynomial equations, the classical B\'ezout bound is $3^N$. For this particular system, we know that the number of solutions is exactly $3^N$ (counting multiplicity) for any parameters $J$ and $\mu^2$ with $\lambda\neq0$. This follows since the system consisting of all the terms of degree three is a decoupled system of monomials. That is, there is only one term of degree three for the $i$th polynomial in \eqref{e:root} which depends only upon $q^\text{s}_i$, namely the monomial $\lambda(q^\text{s}_i)^3/6$. This implies that \eqref{e:root} has no solutions ``at infinity'' so that the classical B\'ezout bound must be sharp (counting multiplicity). Thus, we have a solid check on our claim to find all solutions using homotopy continuation. However, the problem is that $3^N$ grows rapidly as $N$ increases and, due to current computational limitations, we are restricted to only small size lattices such as $3\times 3$ and $4\times 4$.

For the $3\times 3$ lattice, it took an average of roughly a minute to compute the $3^9$ solutions (counting multiplicity) for a given value of $J$ using Bertini running on a 2.4 GHz Opteron 250 processor with 64-bit Linux. For the $4\times 4$ lattice, it took an average of roughly $8.5$ hours to compute the $3^{16}$ solutions (counting multiplicity) for a given value of $J$ using Bertini running on a cluster consisting of $12$ nodes, each containing two 2.33 GHz quad-core Xeon 5410 processors running 64-bit Linux.


\section{Properties of stationary points}
\label{s:properties}
Using the NPHC method as explained in the previous section, we can obtain all complex stationary points of $V$. In the context of energy landscape methods, one is usually interested in the real solutions only, i.e., solutions of \eqref{e:root} with $q^\text{s}\in\RR^N$. In the next few subsections, we report on the properties of these real stationary points: In Sec.\ \ref{s:real} the number of real stationary points is analyzed and the existence of singular solutions is discussed. In Sec.\ \ref{s:values} we study the potential energies $V(q^\text{s})$ of the real $q^\text{s}$, and in Sec.\ \ref{s:Hesse} their Hessian determinants. In Sec.\ \ref{s:Euler} the Euler characteristic of certain submanifolds in configuration space, computed from the indices of the real stationary points, is investigated. Since, as mentioned in the Introduction and discussed in a Letter \cite{KastnerMehta11}, we found that the real stationary points are not related to the phase transition of the model (at least not in the direct way predicted by the theorem in Ref.\ \cite{FraPe04,*FraPeSpi07}), we extended our analysis to include complex stationary points. The results of this analysis are reported in Sec.\ \ref{s:complex}.


\subsection{Real stationary points}
\label{s:real}

For $J=0$, i.e., in the absence of coupling, the stationary points $q^\text{s}$ of the potential $V$ in \eqref{e:V} can be calculated analytically, obtaining $3^N$ distinct solutions $q^\text{s}=(q^\text{s}_1,\dotsc,q^\text{s}_N)$ with $q^\text{s}_j\in\{0,\pm\sqrt{6\mu^2/\lambda}\}$. Since $\lambda,\mu>0$, these stationary points are all real. Upon increasing the coupling constant $J$, real stationary points start to bifurcate into complex ones, and the number of real stationary points decreases gradually from $3^N$ for $J=0$ to only 3 stationary points for some sufficiently large $J$. This behavior is illustrated for $3\times3$ and $4\times4$ lattices in Fig.\ \ref{f:number}. The three stationary points that persist at large $J$ are the two global minima $q^\text{s}=(q^\text{s}_1,\dotsc,q^\text{s}_N)$ where all $q^\text{s}_j=\sqrt{6\mu^2/\lambda}$, respectively $-\sqrt{6\mu^2/\lambda}$, and a stationary point of index $1$ where all $q^\text{s}_j=0$.

\begin{figure}\center
\includegraphics[width=0.49\linewidth]{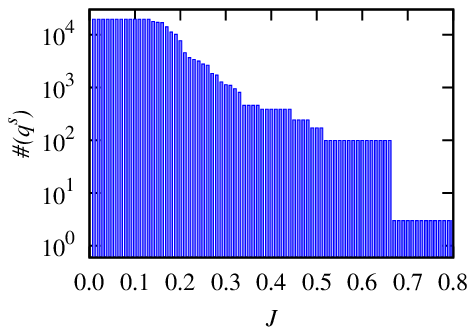}
\includegraphics[width=0.49\linewidth]{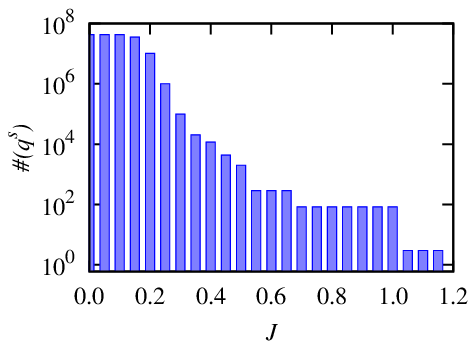}
\caption{\label{f:number}
(Color online) The number of real stationary points of $V$ for $3\times3$ (left) and $4\times 4$ (right) lattices, plotted logarithmically as a function of the coupling $J$.
}
\end{figure}
The value of $J$ at which the number of real solutions drops to 3 can be computed semi-analytically. This is done by computing with {\sc Mathematica} the index of the stationary point $q^\text{s}=(0,\dotsc,0)$ as a function of $J$ and then search for the value of $J$ at which the index drops to 1. Strictly speaking this value of the index does not guarantee that there are indeed only 3 real solutions, but the evidence we find makes it appear plausible at least:
\begin{enumerate}
\item For the $3\times3$ and $4\times4$ lattices where we can compute all stationary points, the index drops to 1 at the same value where the number of solution drops to 3.
\item Once the index is 1, it remains zero for all larger $J$ we tested. Assuming that, as in the case of the $3\times3$ and $4\times4$ lattices, the number of real stationary points always decreases with increasing $J$ and reaches 3 for some value of $J$, it appears plausible that this last change of the index happens when the number of real solution drops to its lowest value of 3.
\item Index 1 of the stationary point $q^\text{s}=(0,\dotsc,0)$ is consistent with the existence of just 3 real solutions from the point of view of the Euler characteristic \eqref{e:Euler} as introduced in Sec.\ \ref{s:Euler}:  The two global minima (having index zero) contribute $+2$ to the Euler characteristic. In the absence of other stationary points, $q^\text{s}=(0,\dotsc,0)$ has to contribute $-1$ which is achieved by a stationary point of index 1 (but any other odd index would have worked as well).
\end{enumerate}
Accepting this reasoning as plausible, we find the values of $J(N)$ at which the number of solutions drops to 3 to be $N$-dependent and to be fitted excellently by a parabola, as shown in Fig.\ \ref{f:JofL}.
\begin{figure}[b]\center
\includegraphics[width=0.7\linewidth]{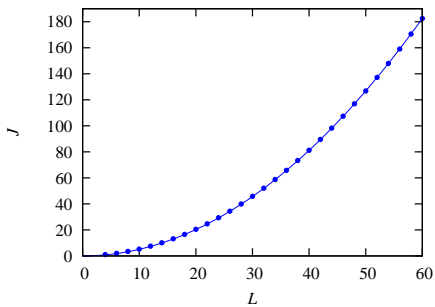}
\caption{\label{f:JofL}
(Color online) The value of $J$ at which, for a given linear system size $L$, the number of real stationary points of $V$ drops to 3. The dots are data points computed with {\sc Mathematica}, the line is the parabola $0.0507366 L^2$ fitted to the data.
}
\end{figure}

We have also investigated the values of $J$ for which the system has at least one real singular solution, i.e., bifurcation points of the parametric systems, using NPHC. At these solutions the potential has degenerate critical points, a feature that does not make $V$ qualified to directly apply Morse theory as described in Section \ref{s:Euler}. There are two approaches that we used to compute where the bifurcations in a one-parameter system occur, which we describe in the context of computing where the first bifurcation occurs. In the first approach, we use the basic philosophy of the NPHC method with a slight change that we treat $J$ itself as a continuation parameter, i.e., we start with the known solutions at $J = 0$ and simply track the solutions as $J$ increases to determine the smallest value of $J>0$ where solutions coalesce. This yielded the values of $J \approx 0.12907$ and $J \approx 0.12894$ for the $3\times 3$ and $4\times 4$ lattice, respectively.

In the second approach, we use the fact that the Hessian determinant, $\det{\mathcal H}_V(q,J)$, where,
\begin{equation}\label{e:Hesse}
{\mathcal H}_V(q) = \left(\frac{\partial^2 V(q)}{\partial q_i \partial q_j}\right)_{\!i,j},
\end{equation}
is zero at the singular solutions. We add this equation, $\det{\mathcal H}_V(q,J) = 0$, as an additional equation in the system of stationary equations leaving $J$ unfixed so that it can be treated as a variable. We then use Bertini to compute the set ${\mathcal S}$ of values of $J$ where this combined system has a solution. Since all of the solutions at $J = 0$ are nonsingular, it follows that the set ${\mathcal S}$ is the set of roots of a nonzero univariate polynomial $s(x)$. In particular, ${\mathcal S}$ is a finite set of points. See Appendix \ref{s:appendix} for more details.

The coefficients of the polynomial $s$ depend upon $\lambda$ and $\mu^2$. If $\lambda$ and $\mu^2$ are rational numbers, then $s$ has rational coefficients meaning that $\mathcal S$ is a finite subset of the set of algebraic numbers, a countable subset of $\CC$. For example, with $\lambda = 3/5$ and $\mu^2 = 2$, we know that the set $V$ of complex stationary points must contain $3^N$ distinct points when $J$ is a transcendental number, e.g., $J = \pi$.

For the $3 \times 3$ lattice, Bertini found that $\mathcal S$ consists of $1357$ complex numbers, of which $297$ are real and $178$ are positive. The smallest positive value using this approach is also $J \approx 0.12907$. This computation also yields that, for $J > 11.00169$, all stationary points must be nonsingular. Performing this same computation using the $4\times 4$ lattice is currently beyond the available computational resources.
\begin{figure*}\center
\includegraphics[width=0.245\linewidth]{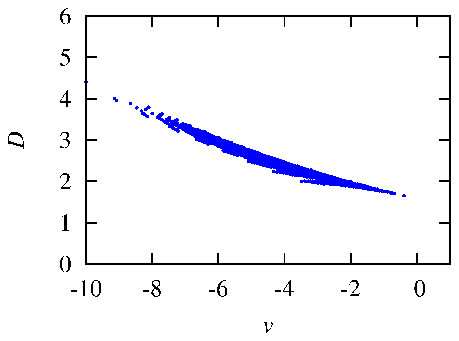}
\includegraphics[width=0.245\linewidth]{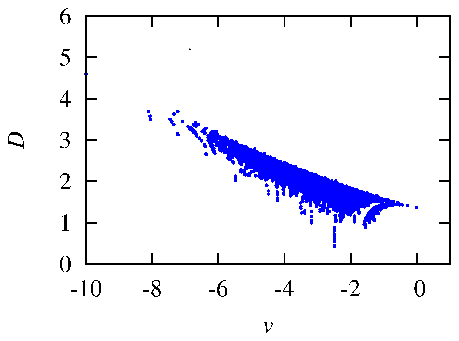}
\includegraphics[width=0.245\linewidth]{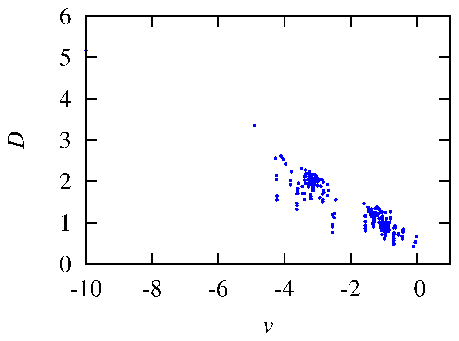}
\includegraphics[width=0.245\linewidth]{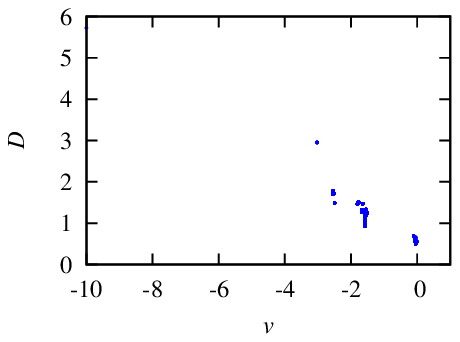}
\caption{\label{f:Hesse}
(Color online) The scaled Hessian determinant $D$ plotted vs.\ the stationary values $v^\text{s}$ for all real stationary points $q^\text{s}$ of a $4\times4$ lattice with couplings $J=0.1$, 0.15, 0.3, and 0.45 (from left to right). To compare with the corresponding values of the phase transition energy $v_\text{c}$, see Fig.\ \ref{f:vc_vs_J}. The distribution of potential energies also illustrates that $v^\text{s}\leq0$ for all $q^\text{s}$, as discussed in Sec.\ \ref{s:values}.
}
\end{figure*}


\subsection{Stationary values}
\label{s:values}

In the Introduction, we briefly reviewed the research efforts aiming at establishing a relation between phase transitions and stationary points of the potential energy function $V$. These efforts all have in common that they focus on a conjectured relation between the occurrence of a phase transition at some critical potential energy $v_\text{c}$ and the properties of stationary points $q^\text{s}$ with stationary values $v^\text{s}=V(q^\text{s})/N$ coinciding with $v_\text{c}$ \footnote{In contrast to other approaches which focus on what is called the {\em underlying stationary points}; see \cite{AnRuZa05,*AnRu08}.}. From the stationary points obtained by means of the numerical homotopy continuation method for lattice sizes $3\times3$ and $4\times4$, it is straightforward to compute, via \eqref{e:V}, the stationary values $v^\text{s}$. For arbitrary couplings $J$, we found that $v^\text{s}\leq0$ for all stationary points $q^\text{s}$. An analytical calculation, reported in Ref.\ \cite{KastnerMehta11}, has confirmed this observation and extended it to lattices of arbitrary sizes. As explained in this same reference, it is this upper bound on $v^\text{s}$ which disproves the theorem by Franzosi and Pettini \cite{FraPe04,*FraPeSpi07}, as it cannot be reconciled with the fact that the critical energy $v_\text{c}$ of the phase transition becomes positive for couplings $J\gtrsim0.7$.


\subsection{Hessian determinant}
\label{s:Hesse}

Once a relation between stationary points of the potential energy landscape and the occurrence of phase transitions had been conjectured in the 1990s, it immediately became clear that not all stationary points induce phase transitions. Therefore an obvious question to ask was: Is there a certain property of a stationary point that renders it capable of inducing a phase transition? Some years later it was noticed that the Hessian determinant ${\mathcal H}_V$ of the potential energy function $V$, evaluated at the stationary points, is crucial for discriminating whether or not a stationary point can induce a phase transition in the thermodynamic limit \cite{KaSchneSchrei07,*KaSchne08,*KaSchneSchrei08}. For some models, even in the absence of an exact solution, this insight facilitated the exact analytic computation of transition energies \cite{NardiniCasetti09,*Kastner11}. We refrain here from stating the precise criterion, noting only that stationary points with a Hessian determinant approaching zero in the thermodynamic limit play an important role.

We evaluated the determinant of the Hesse matrix \eqref{e:Hesse} at all of the real stationary points $q^\text{s}$ of $V$ obtained by the homotopy continuation method. In Fig.\ \ref{f:Hesse}, we show the rescaled Hessian determinant
\begin{equation}\label{e:D}
D=|\det{\mathcal H}_V(q^\text{s})|^{1/N},
\end{equation}
plotted versus the stationary values $v^\text{s}=V(q^\text{s})/N$ for all real stationary points of $4\times4$ lattices and various couplings $J$. From these plots one can immediately verify that $v^\text{s}\leq0$ for all real stationary points and arbitrary coupling $J$, as discussed in Sec.\ \ref{s:values}. Since in general (i.e., at least for sufficiently large $J$) the potential energy at which the phase transition occurs is not close to any of the stationary points, there is no point in discussing the Hessian determinant as a possible signature of the transition in the spirit of what was proposed in the abovementioned references \cite{KaSchneSchrei07,*KaSchne08,*KaSchneSchrei08}. In Sec.\ \ref{s:NR} we will use the data as presented in Fig.\ \ref{f:Hesse} for a different purpose, namely to compare the homotopy continuation data to those obtained by means of the Newton-Raphson method.
\begin{figure*}\center
\includegraphics[width=0.245\linewidth]{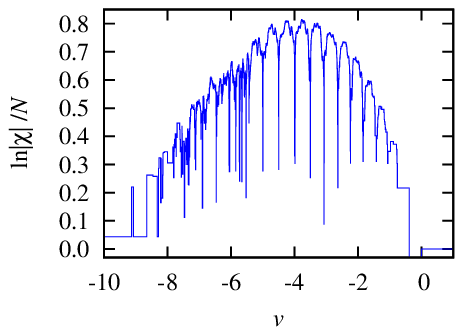}
\includegraphics[width=0.245\linewidth]{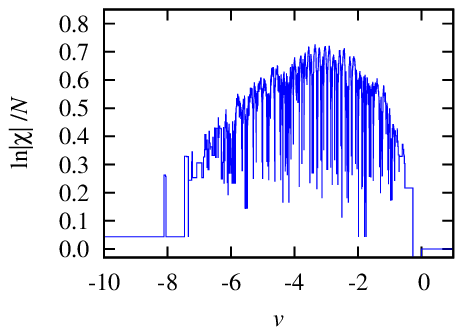}
\includegraphics[width=0.245\linewidth]{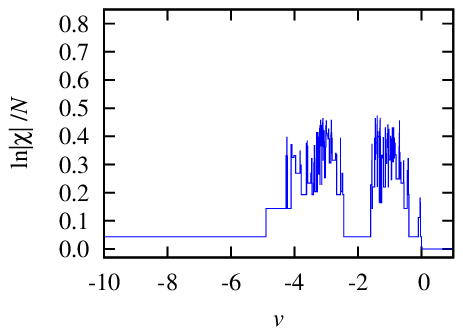}
\includegraphics[width=0.245\linewidth]{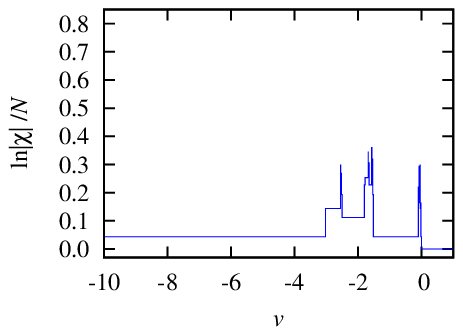}
\caption{\label{f:euler_char1}
(Color online) Graphs of the logarithm of the Euler characteristic, $\ln|\chi(M_v)|/N$, for a $4\times4$ lattice and coupling strengths $J=0.1$, 0.15, 0.3, and 0.45 (from left to right). Note that these results are exact and the oscillations visible are not a consequence of noisy data.}
\end{figure*}


\subsection{Euler characteristic}
\label{s:Euler}

In the Introduction, and also at the beginning of Sec.\ \ref{s:properties}, we referred to the work of Franzosi and Pettini \cite{FraPe04,*FraPeSpi07} or to related publications as dealing with the relation of stationary points of the potential energy function $V$ to thermodynamic phase transitions. Although this is correct as regards content, it is not obvious at first glance, as these results were originally phrased in terms of topology changes of certain submanifolds $M_v$ in configuration space $\Gamma$,
\begin{equation}
M_v=\left\{q\in\Gamma\,\big|\,V(q)\leq Nv\right\}.
\end{equation}
Upon variation of the parameter $v$, the topology of the submanifolds $M_v$ may change at some value $v_\text{t}$, in the sense that $M_v$ is not homeomorphic to $M_w$ for $v<v_\text{t}$ and $w>v_\text{t}$. The occurrence of phase transitions at some critical potential energy $v_\text{c}$ was then conjectured to be related to the presence of topology changes with energies $v_\text{t}$ in an open neighborhood of $v_\text{c}$. Via Morse theory, such topology changes can be related to the presence of stationary points of $V$ with stationary values $v^\text{s}=v_\text{t}$ (see \cite{Kastner08} for an elementary introduction or \cite{Matsumoto,*Milnor} for a textbook presentation).

In the context of configuration space topology, the Euler characteristic $\chi(M_v)$ of the manifolds $M_v$ has been used in several publications as a way of characterizing the changes of topology \cite{FraPeSpi00,CaCoPe02,*CaPeCo03,*Angelani_etal03}. The Euler characteristic $\chi$ is a topological invariant, i.e., different values of $\chi$ for manifolds $M_v$ and $M_w$ imply that $M_v$ and $M_w$ are not homeomorphic. Hence monitoring the Euler characteristic of the family $\left\{M_v\right\}_{v\in\RR}$ of configuration space subsets under variation of the parameter $v$, we may get an impression of the way the topology of the $M_v$ changes. Plotting the related quantity
\begin{equation}
\sigma(v)=\lim_{N\to\infty}\frac{1}{N}\ln|\chi(M_v)|
\end{equation}
as a function of the potential energy $v$, a kink in $\sigma$ was observed precisely at the critical energy $v_\text{c}$ of the phase transition for several models studied \cite{FraPeSpi00,CaCoPe02,*CaPeCo03,*Angelani_etal03,*Kastner11,*MehtaKastner}.

Knowing all stationary points of $V$ with stationary values $v^\text{s}$ up to a given value $v$, the Euler characteristic of $M_v$ can be calculated by means of the formula
\begin{equation}\label{e:Euler}
\chi(M_v)=\sum_{i=0}^N (-1)^i \mu_i(v),
\end{equation}
where the Morse numbers $\mu_i(v)$ are defined in this context as the number of stationary points $q^\text{s}$ of $V$ with index $i$ and stationary value $v^\text{s}\leq v$. The index $i$ is defined as the number of negative eigenvalues of the Hessian matrix ${\mathcal H}_V(q^\text{s})$, which is assumed to have only nonzero eigenvalues. As we noted earlier, for finitely many values of $J$ the corresponding systems of equations indeed possess singular solutions (see the appendix). Using the NPHC method, we know which of the values of $J$ possess at least one singular solution and in this section, we avoid such values of $J$.

We have computed the Euler characteristic $\chi(M_v)$ from the real stationary points $q^\text{s}$ of $V$ as obtained by the homotopy continuation method, and the results are plotted as a function of $v$ and for various values of $J$ in Fig.\ \ref{f:euler_char1}. 
Since the energy levels are very closely spaced, it is difficult to distinguish one from another. Here, we use the tolerance $10^{-8}$, i.e., if $|v_1 - v_2|\geq 10^{-8}$, then $v_1$ and $v_2$ are distinct energy levels. No kink or other signature is visible in $\chi(M_v)$ at $v=v_\text{c}$: As was discussed in Sec.\ \ref{s:values}, the stationary values $v^\text{s}$ are nonpositive, and the Euler characteristic $\chi(M_v)$ is therefore constant for $v>0$. 
Since the critical potential energy $v_\text{c}$ is positive for $J\gtrsim0.7$, it is clear that in this case the phase transition cannot be signaled by a signature in $\chi(M_v)$. For $J\lesssim0.7$ our results cannot exclude that a signature of the phase transition is somewhere hidden in the data of $\chi(M_v)$, but it seems unlikely that above and below this seemingly arbitrary value of $J$ the behavior should be so different.

Note that a similar conclusion holds for the Euler characteristic $\chi(\Sigma_v)$ of the related manifolds
\begin{equation}
\Sigma_v=\left\{q\in\Gamma\,\big|\,V(q)= Nv\right\}.
\end{equation}
These submanifolds of $\Gamma$ are the boundaries of $M_v$, and their Euler characteristic has been studied numerically in Ref.\ \cite{FraPeSpi00}. For potential energies $v>0$, we know that the manifold $M_v$ is homeomorphic to an $N$-dimensional ball. Its boundary is therefore homeomorphic to an $(N-1)$-sphere, and its Euler characteristic is constant for $v>0$. In fact we have $\chi(\Sigma_v)=0$ or 2, depending on whether $N$ is odd or even. Comparing this exact result with the plot of numerical data in Fig.\ 3 of \cite{FraPeSpi00}, we have to conclude that the behavior of $\chi(\Sigma_v)$ reported in this reference must be an artefact of the numerical method employed.

We can use the results of our computation of the Euler characteristic $\chi(M_v)$ as a consistency check: For potential energies $v>0$ where the manifold $M_v$ is homeomorphic to an $N$-dimensional ball, the Euler characteristic is known to be $\chi(M_v)=1$ for all $v>0$. Computing the alternating sum \eqref{e:Euler} with all the stationary points and their indices as an input, we find that at $v=0$, $\chi(M_v)=1$. Since there is no stationary point for $v>0$, $\chi(M_v)=1$ for all $v>0$, as it should be. We have confirmed this result for all the values of $J$ without singular solutions used in this paper.


\subsection{Complex stationary points}
\label{s:complex}

In Sec.\ \ref{s:values}, we discussed the fact that, for arbitrary coupling $J$, the stationary values $v^\text{s}$ are never positive, while the critical energy $v_\text{c}$ of the phase transition of the nearest-neighbor $\phi^4$ model becomes positive for $J\gtrsim0.7$. A direct relation between phase transitions and stationary points of $V$ (in the spirit of the one in Ref.\ \cite{FraPe04,*FraPeSpi07}) is hence ruled out, but one might wonder if a modification of the conjectured relation might be more successful.

One possible and rather straightforward generalization of this conjecture is obtained by considering not only real stationary points, but also complex ones. The reasoning behind this generalization is that the presence of complex stationary points whose imaginary parts go to zero with increasing system size $N$ should have the same (or at least a similar) effect on the thermodynamic properties of the system as their real counterparts. To test this idea, we have used the (in general complex) stationary points $q^\text{s}$ obtained by means of the homotopy continuation method and plotted in Fig.\ \ref{f:complex} real and imaginary parts of the (complex) potential $V(q^\text{s})$ for various values of the coupling $J$.

\begin{figure*}\center
\includegraphics[width=0.245\linewidth]{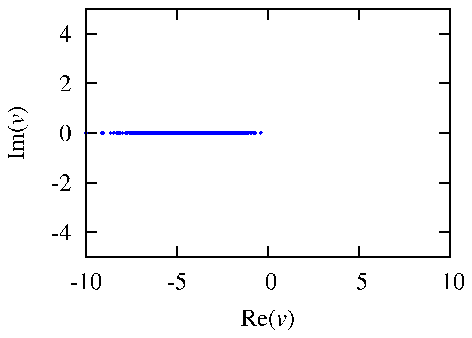}
\includegraphics[width=0.245\linewidth]{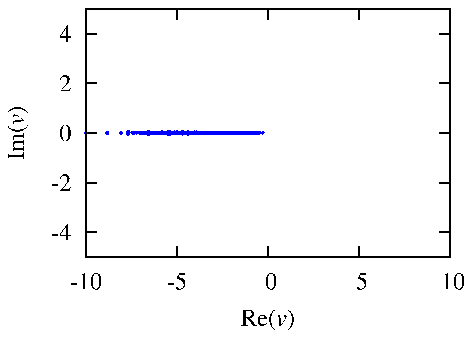}
\includegraphics[width=0.245\linewidth]{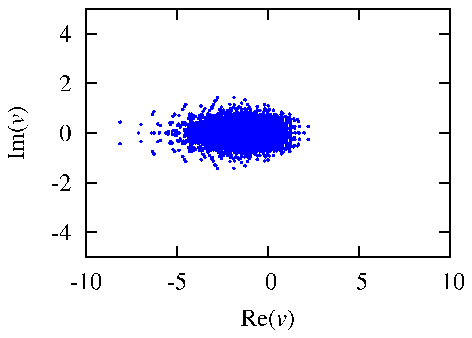}
\includegraphics[width=0.245\linewidth]{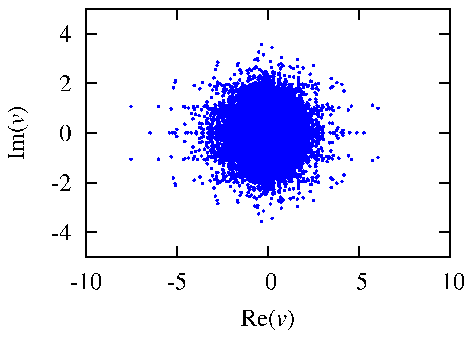}
\caption{\label{f:complex}
(Color online) Imaginary vs.\ real parts of the complex potential energies $V(q^\text{s})/N$ for all complex solutions $q^\text{s}$ of a $4\times 4$ lattice with coupling $J=0.1$, 0.15, 0.3, and 0.45 (from left to right). For small couplings $J\lesssim0.2$ the real part is nonpositive for all $q^\text{s}$, whereas for larger couplings some of the stationary values move into the right halfplane.}
\end{figure*}
At first sight the results are encouraging, as they show that, for sufficiently large $J$, there exist complex $q^\text{s}$ with positive real stationary values $V(q^\text{s})$. Moreover, for the couplings $J$ we studied, the maximal real stationary value is larger than the critical potential energy of the phase transition. Unfortunately, from the data we have there is not much more we can say, and it would be unreasonable to conjecture a relation of the above mentioned kind on the basis of our results.


\section{Newton-Raphson method}
\label{s:NR}

The Newton-Raphson method is a powerful and frequently used iterative algorithm for approximating the roots of a function (see Sec.\ 9.7 of \cite{NumRecC}). In the context of energy landscapes, the stationary points of $V$ are determined by the system of $N$ equations \eqref{e:root}, so the problem is equivalent to finding the roots of the vector-valued function on the left-hand side of \eqref{e:root}.

From a given initial point in phase space, the Newton-Raphson method iteratively finds approximations to a stationary point. If the function has more than one stationary point, it will depend on the initial value of the iteration which of the stationary points is found. For the potential energy function \eqref{e:V} of the two-dimensional nearest-neighbor $\phi^4$ model, we have seen in Sec.\ \ref{s:homotopy} that, at least for small coupling $J$, the number of stationary points is exponentially large in the number $N$ of lattice sites. The result of the Newton-Raphson computation will therefore crucially depend on the set of initial points chosen for the iterations. First, the initial points have to differ sufficiently from each other, in order to make sure that different stationary points are found in the various iteration runs. Second, properties of the initial points will have an influence on the properties of the stationary points found, as the outcome of a Newton-Raphson run typically yields a stationary point that is in some sense close to the initial point.

For a given coupling $J$ and lattice sizes up to $N=32\times32$, we generated sets of $10^6$ initial points by means of a standard Metropolis Monte Carlo dynamics in configuration space \cite{MeRoRoTeTe53,*KastnerMC}. The temperature $T$ in the canonical acceptance rate of the Monte Carlo algorithm was set to $T=100$, and we will comment on this choice of $T$ towards the end of this section. Starting from each of the thus generated initial points, the routine {\tt newt} from \cite{NumRecC}, a globally convergent version of the Newton-Raphson method, was used to compute stationary points of $V$. Like in the homotopy continuation computations, all stationary points $q^\text{s}$ were found to have nonpositive potential energies $v^\text{s}\leq0$, and the number of stationary points was found to decrease dramatically with increasing coupling $J$.

For smaller couplings ($J=0.1$ and $J=0.2$) where the number of stationary points is large, we have plotted the results of the Newton-Raphson calculations in Fig.\ \ref{f:NR}.%
\begin{figure*}\center
\includegraphics[width=0.245\linewidth]{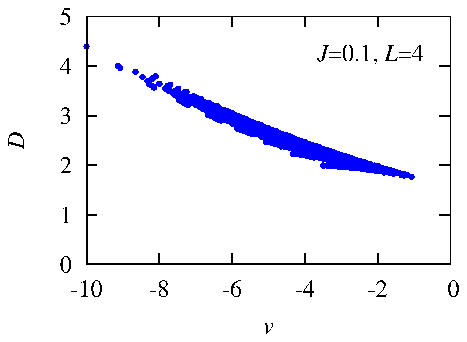}
\includegraphics[width=0.245\linewidth]{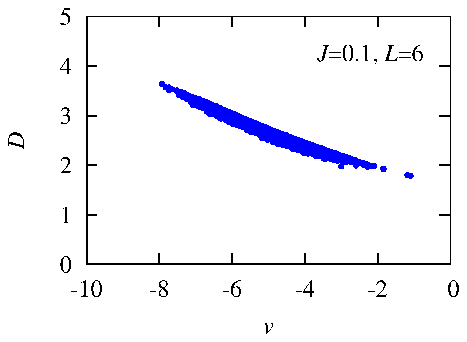}
\includegraphics[width=0.245\linewidth]{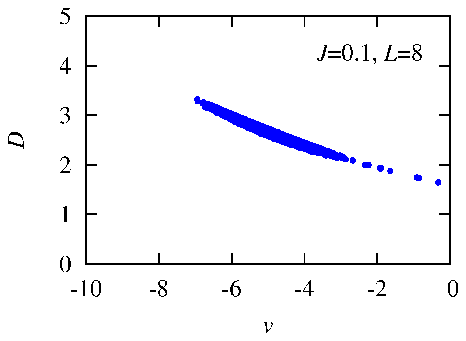}
\includegraphics[width=0.245\linewidth]{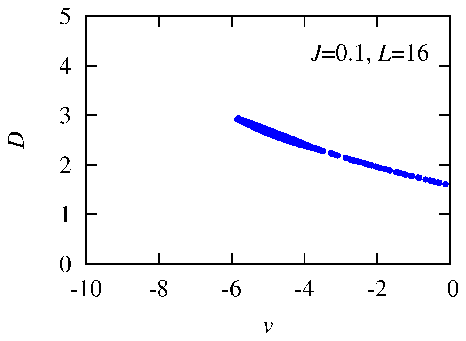}
\includegraphics[width=0.245\linewidth]{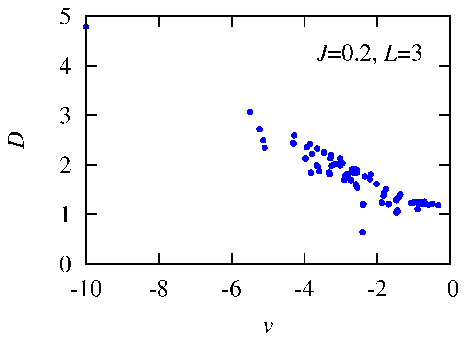}
\includegraphics[width=0.245\linewidth]{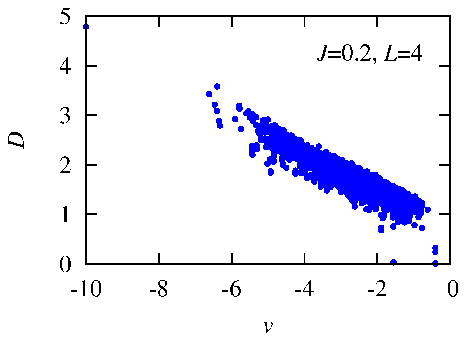}
\includegraphics[width=0.245\linewidth]{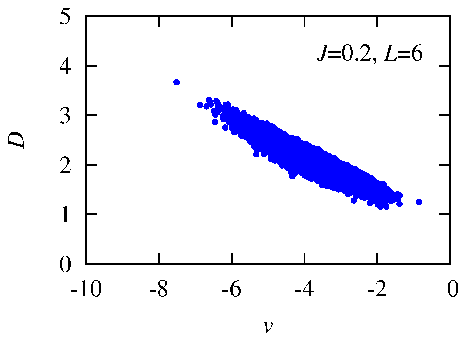}
\includegraphics[width=0.245\linewidth]{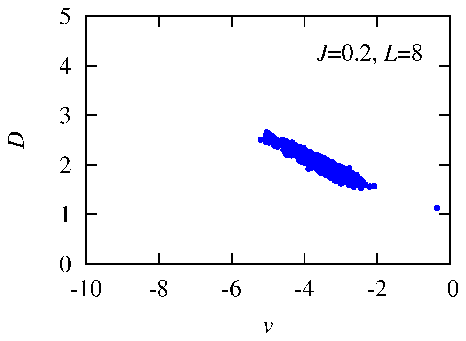}
\caption{\label{f:NR}
(Color online) Numerical results from the Newton-Raphson method. For system sizes $N=L\times L$ with $L=3$, 4, 6, 8, and 16, the scaled Hessian determinant $D$ is shown versus the stationary value $v^\text{s}$. Up to $10^6$ different stationary points $q^\text{s}$ per system size have been computed for $J=0.1$ (top row) and $J=0.2$ (bottom row).
}
\end{figure*}
Like for the results from numerical continuation in Sec.\ \ref{s:homotopy}, we have plotted the scaled Hessian determinant $D$ at a stationary point versus its stationary value $v^\text{s}$. For the smaller system sizes $N=L\times L$ with $L=3$ and $L=4$, the shapes of the clouds of points shown in Fig.\ \ref{f:NR} resemble the ones produced from the complete set of stationary points in Fig.\ \ref{f:Hesse}. For larger system sizes $L=6$, 8, 16, the cloud of points becomes more and more focused, being localized in that region of the $(v,D)$ plane where the concentration of stationary points is largest.

We have seen that, in contrast to the homotopy continuation method where only small system sizes $L=3$ and $L=4$ were accessible, the Newton-Raphson method can be applied to much larger sizes up to $L=32$ (and even larger with more numerical effort). However, for small couplings $J$ and the larger $L$ considered, the number of real stationary points of $V$ is expected to be of the order of $3^N$, and it is evident that we can not compute more than a small fraction of them.

This is reminiscent of the situation one encounters in Monte Carlo simulations where only a tiny subset of a tremendously large configuration space can be sampled. In the Monte Carlo context, the problem can be overcome (or at least significantly abated) by the technique of importance sampling \cite{KastnerMC}. We have tried a very straightforward (and possibly naive) adaptation of this idea to the Newton-Raphson computation of stationary points, simply by adjusting the parameter $T$ of the Metropolis importance sampling algorithm which was used for generating the initial points of the Newton-Raphson search. Somewhat disappointingly, the shape of the cloud of points in Fig.\ \ref{f:NR} turned out to be entirely insensitive to changes in $T$. Using for example a small value of $T$, we would have expected to end up with stationary points of lower potential energy on average, but surprisingly this was not the case.

There are other, more involved ways of how one could shift the search of stationary points to higher or lower potential energies, but we have not yet implemented such refinements. One could, for example, use a more advanced search routine (like the OPTIM program package \cite{WalesOPTIM}) which allows one to search for stationary points of a given index, i.e., of a given number of negative eigenvalues of the Hessian at the stationary point. Since the index of a stationary point and its potential energy are expected to be correlated, such a routine should find stationary points of low energy when searching for small indices, and {\em vice versa}.

\section{Conclusions}
\label{s:conclusions}

Two numerical methods for the computation of stationary points of multivariate functions were discussed in this article: the numerical polynomial homotopy continuation method (NPHC) and a globally-convergent variant of the Newton-Raphson method. We applied both methods to the potential energy function $V$ of the two-dimensional nearest-neighbor $\phi^4$ model on $L\times L$ square lattices. The NPHC method allows one to obtain all stationary points of $V$, but is limited to system sizes up to $4\times4$ with the computational resources we had at our disposal. With the Newton-Raphson method we have computed stationary points for larger lattices of up to $32\times32$ sites, but only a small subset of all the stationary points of such a large system could be obtained.

The motivation for this type of study originates from a number of conjectures relating the stationary points of $V$ to the occurrence of phase transitions in the thermodynamic limit. These conjectures refer to certain quantities which can be computed from the stationary points of $V$, like their potential energies, their Hessian determinants, and the Euler characteristic of the underlying potential energy manifolds in configuration space. We have calculated these and a few other quantities from the stationary points of the $\phi^4$ model obtained with NPHC and Newton-Raphson, but---contrary to what the conjectures suggest---no sign of the phase transition of the model was found. This failure and its consequences, including the falsification of a theorem allegedly proven in Ref.\ \cite{FraPe04,*FraPeSpi07}, was discussed in a Letter \cite{KastnerMehta11}.

The NPHC results for the nearest-neighbor $\phi^4$ model on a $4\times4$ lattice can be overviewed as follows:
\begin{enumerate}
\item The number of real stationary points decreases from $3^N$ for $J=0$ to only 3 with increasing $J$
\item For any finite $N$, singular solutions occur only for finitely many values of $J$.
\item The stationary values $v^\text{s}$ are all nonpositive for arbitrary couplings $J$.
\item The Euler characteristic, computed as the alternating sum of the Morse numbers, confirms the correct and complete computation of all the stationary points.
\item Unlike real stationary points, complex stationary points of $V$ can have positive stationary values, but we were unable to identify a relation between these positive values and the positive phase transition energy of the $\phi^4$ model for larger $J$.
\end{enumerate}

Since the Newton-Raphson method yields only a subset of all the stationary points, we compared these results for system sizes up to $16\times16$ to those obtained by the NPHC method for $4\times4$ lattices. For this comparison we chose plots of the rescaled Hessian determinant $D$ as defined in \eqref{e:D} vs.\ the potential energy $v$. A comparison of different lattice sizes is of course problematic, but a general trend can be deduced: For system sizes $8\times8$ and larger, the number of stationary points becomes in general so large that only that region in the $(D,v)$-plane is explored where the (strongly peaked) density of stationary points is the highest. Importance sampling may provide a way out of these difficulties, but we have not yet implemented such a scheme. 


\begin{acknowledgments}
D.M.\ acknowledges support by the U.S.\ Department of Energy under contract DE-FG02-85ER40237 and by the Science Foundation Ireland grant 08/RFP/PHY1462.
J.D.H.\ acknowledges support by the U.S.\ National Science Foundation under grants DMS-0915211 and DMS-1114336.
M.K.\ acknowledges support by the {\em Incentive Funding for Rated Researchers programme}\/ of the National Research Foundation of South Africa.
\end{acknowledgments}

\appendix

\section{Morse property of the potential energy}
\label{s:appendix}

In Sec.\ \ref{s:real}, for given $N$, $\lambda$ and $\mu^2$, we considered values of $J$ for which the potential $V$ in \eqref{e:V} is not a proper Morse function, i.e., for which not all of the stationary points of $V$ have a nonzero Hessian determinant. The following provides more details regarding such values of~$J$ and relationship to a theorem by Franzosi and Pettini.

Let ${\mathcal A(N,\lambda,\mu^2)}$ be the set of pairs $(q,J)$ such that $q$ is a singular stationary point (either real or complex) of $V$, i.e., \eqref{e:root} holds and $\det{\mathcal H}_V(q,J) = 0$. Denote by
\begin{equation}
{\mathcal S(N,\lambda,\mu^2)} = \{J\in\CC~|~(q,J)\in{\mathcal A(N,\lambda,\mu^2)} \text{~for some~} q\}
\end{equation}
the set of values $J$ such that the system describing the set of stationary points of $V$ has at least one singular solution.
\begin{proposition}\label{prop:finite}
For each $N\geq 2$ and nonzero $\lambda,\mu\in\CC$, the set ${\mathcal S(N,\lambda,\mu^2)}$ is a finite subset of $\CC$.
\end{proposition}
\begin{proof}
The set ${\mathcal A(N,\lambda,\mu^2)}$ is an {\em algebraic set} and the set ${\mathcal S(N,\lambda,\mu^2)}$ is a {\em constructible algebraic set} (see Chapter 12 of \cite{SW:95}). Lemma~12.5.3 of \cite{SW:95} yields that there is a univariate polynomial $s_{N,\lambda,\mu^2}(x)$ such that the set of roots of $s_{N,\lambda,\mu^2}(x)$ is the closure of ${\mathcal S(N,\lambda,\mu^2)}$ in the complex topology. Since a univariate polynomial is either zero or has finitely many roots, this implies that ${\mathcal S(N,\lambda,\mu^2)}$ is either dense in $\CC$ or is a finite set.  Since all stationary points for $J = 0$ are nonsingular, the Inverse Function Theorem (see Chapter 3 of \cite{DK04}) implies that this must hold in an open neighborhood of $0$. In particular, ${\mathcal S(N,\lambda,\mu^2)}$ can not be dense in $\CC$ and thus must be finite.
\end{proof}

It follows from elimination theory (see Chapter 14 of \cite{Eisenbud}) that the coefficients of $s_{N,\lambda,\mu^2}(x)$ are polynomials in $\lambda$ and $\mu^2$ with rational coefficients. In particular, if $\lambda$ and $\mu^2$ are rational, then $s_{N,\lambda,\mu^2}(x)$ has rational coefficients so that ${\mathcal S(N,\lambda,\mu^2)}$ consists of finitely many algebraic numbers.

\begin{corollary}\label{cor:final}
For nonzero $\lambda,\mu\in\CC$, the set
\begin{equation}
{\mathcal T}(\lambda,\mu^2) = \bigcup_{N\geq2} {\mathcal S(N,\lambda,\mu^2)}.
\end{equation}
is a countable subset of $\CC$.
\end{corollary}
\begin{proof}
It follows from Proposition \ref{prop:finite} that ${\mathcal T(\lambda,\mu^2)}$ is a countable union of finite sets and is therefore countable.
\end{proof}

Corollary~\ref{cor:final} shows that for given $\lambda$ and $\mu^2$, there exist uncountably infinite many values of $J$, which densely cover the real axis, such that the potential energy function \eqref{e:V} is a Morse function. The potential energy function \eqref{e:V} of the nearest-neighbor $\phi^4$ model is therefore a valid counterexample disproving the theorem announced by Franzosi and Pettini in Ref.\ \cite{FraPe04}.

\bibliography{Phi4Long}

\end{document}